\documentclass[12pt, a4paper]{article}

\usepackage{latexsym}
\usepackage{mathrsfs}
\usepackage{amsmath}
\usepackage{amsthm}
\usepackage{array}
\newtheorem{theorem}{Theorem}

\newtheorem{lemma}{Lemma}
\newtheorem{prop}{Proposition}

\newtheorem{corollary}{Corollary}

\usepackage{pst-all}
\usepackage{mathrsfs}

\begin{document}

\title{The vulnerability of the diameter of enhanced hypercubes}

\author
{{\large Meijie Ma$^{a,b,}$\footnote{Corresponding author:
mameij@mail.ustc.edu.cn} \quad Douglas B. West$^{b,c}$ \quad
Jun-Ming Xu$^{d}$ }\\
{\small $^a$ School of Management Science and Engineering}\\
{\small Shandong Institute of Business and Technology, Yantai 264005, China}\\
{\small $^b$Department of Mathematics}\\
{\small Zhejiang Normal University, Jinhua 321004, China}\\
{\small $^c$Department of Mathematics}  \\
{\small University of Illinois, Urbana, IL 61801, USA}\\
{\small $^d$School of Mathematical Sciences}\\
{\small  University of Science and Technology of China, Hefei 230026, China}\\
}

\date{}
\maketitle

\def\esub{\subseteq}
\def\VEC#1#2#3{#1_{#2},\ldots,#1_{#3}}
\def\fktwo{\lfloor \frac k2\rfloor}
\def\cktwo{\lceil \frac k2\rceil}

\setlength{\baselineskip}{22pt}

\begin{abstract}
For an interconnection network $G$, the {\it $\omega$-wide diameter}
$d_\omega(G)$ is the least $\ell$ such that any two vertices are
joined by $\omega$ internally-disjoint paths of length at most
$\ell$, and the {\it $(\omega-1)$-fault diameter} $D_{\omega}(G)$ is
the maximum diameter of a subgraph obtained by deleting fewer than
$\omega$ vertices of $G$.

The enhanced hypercube $Q_{n,k}$ is a variant of the well-known
hypercube.  Yang, Chang, Pai, and Chan gave an upper bound for
$d_{n+1}(Q_{n,k})$ and $D_{n+1}(Q_{n,k})$ and posed the problem of
finding the wide diameter and fault diameter of $Q_{n,k}$. By
constructing internally disjoint paths between any two vertices in
the enhanced hypercube, for $n\ge3$ and $2\le k\le n$ we prove that
$D_\omega(Q_{n,k})=d_\omega(Q_{n,k})=d(Q_{n,k})$  for $1 \leq \omega
< n-\lfloor\frac{k}{2}\rfloor$;
$D_\omega(Q_{n,k})=d_\omega(Q_{n,k})=d(Q_{n,k})+1$ for
$n-\lfloor\frac{k}{2}\rfloor \leq \omega \leq n+1$, where
$d(Q_{n,k})$ is the diameter of $Q_{n,k}$. These results mean that
interconnection networks modelled by enhanced hypercubes are
extremely robust.
\end{abstract}


\section{Introduction}

An interconnection network is conveniently represented by an
undirected graph. The vertices (or edges) of the graph represent the
nodes (or links) of the network. Throughout this paper, vertex and
node, edge and link, graph and network are used interchangeably.
Reliability and efficiency are important criteria in the design of
interconnection networks.  In the study of fault-tolerance and
transmission delay of networks, wide diameter and fault diameter are
important parameters that have been studied by many researchers.
They combine connectivity with diameter to measure simultaneously
the fault-tolerance and efficiency of parallel processing computer
networks. These parameters were studied by several authors for some
Cartesian product graphs \cite{ez2013, x04, xxh05} and for the
hypercube and its variants \cite{c10, de97, kllht09,l1993,
sl12,sy1997, tct2014}.

Let $u$ and $v$ be two vertices in a network $G$.  A {\it $u,v$-path} is a path
with endpoints $u$ and $v$.  The {\it distance} between $u$ and $v$, denoted by
$d(u,v)$, is the minimum length (number of edges) of a $u,v$-path.
The {\it diameter} of $G$, denoted by $d(G)$, is the maximum distance between
vertices.  The {\it connectivity} $\kappa(G)$ is the minimum number of
vertices whose removal results in a disconnected or $1$-vertex network.
We say that $G$ is {\it $k$-connected} when $0 < k \leq \kappa(G)$.  By
Menger's Theorem~\cite{Menger}, in a $k$-connected network there exist $k$
internally disjoint paths joining any two vertices ({\it internally disjoint}
means that the only shared vertices are the endpoints).

Given a $k$-connected graph $G$, fix $\omega$ with $1\le\omega\le k$.
The {\it $\omega$-wide diameter} of $G$, denoted by $d_\omega(G)$, is the least
$\ell$ such that for any $u,v\in V(G)$ there exist $\omega$ internally disjoint
$u,v$-paths of length at most $\ell$.  Throughout this paper, we abuse
terminology by writing ``disjoint paths'' to mean ``internally disjoint paths''.
Note that $d_1(G)$ is just the diameter $d(G)$ of $G$.  From the definition,
$$d(G)=d_1(G)\leq d_2(G)\leq \cdots \leq d_{k-1}(G)\leq d_k(G).
$$

Failures are inevitable when a network is put in use. Therefore, it is
important to consider faulty networks.  The {\it $(\omega-1)$-fault diameter}
of a graph $G$, denoted by $D_\omega(G)$, is the maximum diameter among
subgraphs obtained from $G$ by deleting fewer than $\omega$ vertices; it
measures the worst-case effect on the diameter when vertex faults occur.
Note that $D_\omega(G)$ is well-defined if and only if
$G$ is $\omega$-connected, moreover,
$$d(G)=D_1(G)\leq D_2(G)\leq \cdots \leq D_{k-1}(G)\leq D_k(G). $$

From the definitions, it follows that $D_\omega(G)\leq d_\omega(G)$
when $G$ is $k$-connected and $1\le\omega\le k$~\cite{lcch98}.
Equality holds for some well-known networks~\cite{dhl96,lc99}.

As a topology for an interconnection network of a multiprocessor
system, the hypercube is a widely used and well-known
model, since it possesses many attractive properties
such as regularity, symmetry, logarithmic diameter, high
connectivity, recursive construction, ease of bisection, and
relatively low link complexity~\cite{l92,ss88,x01}.  We study an important
variant of the hypercube $Q_n$, the enhanced hypercube $Q_{n,k}$ proposed by
Tzeng and Wei~\cite{tw1991}; its properties have been studied
in~\cite{clthh2009, l08, w1994, w1999, ycpc2014}. We give the definition and
basic properties of $Q_{n,k}$ in Section 2.

It was shown by Liu~\cite{l08} that $\kappa (Q_{n,k})=n+1$. Thus,
the wide diameter $d_\omega(Q_{n,k})$ and the fault diameter
$D_\omega(Q_{n,k})$ are well-defined when $\omega \leq n+1$.
Yang, Chang, Pai, and Chan~\cite{ycpc2014} gave an upper bound for
$d_{n+1}(Q_{n,k})$ and $D_{n+1}(Q_{n,k})$, and they posed the problem of
finding the wide diameter and fault diameter of $Q_{n,k}$. In this
paper, for $n\ge3$ and $2\le k\le n$, we prove
$$
D_\omega(Q_{n,k})=d_\omega(Q_{n,k})=\begin{cases}
d(Q_{n,k}) & \textrm{for $1 \leq \omega < n-\lfloor\frac{k}{2}\rfloor$;}\\
d(Q_{n,k})+1 & \textrm{for $n-\lfloor\frac{k}{2}\rfloor \leq \omega \leq n+1$.}
\end{cases}
$$
The special case $k=n$ (folded hypercube) was obtain earlier by
Sim\'o and Yebra~\cite{sy1997}, along with the same values for edge deletions.
For enhanced hypercubes also, our arguments yield the same values for
edge deletions as for vertex deletions.

\section{Properties of $Q_{n,k}$}

Let $x_n\cdots x_1$ be an $n$-bit binary string.  We call the rightmost bit
the first bit and the leftmost bit the $n$th bit. For simplicity we use
$a^i$ to mean that the bit $a$ is repeated $i$ times;
for example, $01^30^2 = 011100$.  The {\it Hamming distance}
between strings $u$ and $v$, denoted by $H(u,v)$, is the number of
positions where the two strings differ.

The {\it $n$-dimensional  hypercube} $Q_n$ is the graph whose vertices are the
$n$-bit binary strings and whose edges are the pairs of vertices differing in
exactly one position.  An edge of $Q_n$ is a {\it $j$-dimensional edge} if the
two endpoints differ in the $j$th position.  For $1\le j\le n$, let $E_j$
denote the set of $j$-dimensional edges in $Q_n$.

As a variant of the hypercube, the {\it $n$-dimensional folded hypercube}
$FQ_n$, proposed first by El-Amawy and Latifi~\cite{as91}, is obtained from the
hypercube $Q_n$ by making each vertex $u$ adjacent to its complementary vertex,
denoted $\bar u$ and obtained from $u$ by subtracting each bit from $1$.
Such an edge is often called a {\it complementary edge}.

For $2\leq k\leq n$, the {\it $n$-dimensional enhanced hypercube} $Q_{n,k}$ is
obtained from the hypercube $Q_n$ by adding the edge $uv$ whenever $u$ and $v$
are related by $u=x_n\cdots x_1$ and
$v=x_n\cdots x_{k+1}\bar{x}_k\bar{x}_{k-1}\cdots \bar{x}_1$; that is, the
first $k$ bits are complemented.
Such an edge is called a {\it $k$-complementary edge}.
For convenience, we use $E_{0}$ to denote the set of $k$-complementary edges.
Thus $E(Q_{n,k})=E(Q_n)\cup E_{0}$. When $k=n$, we have $Q_{n,n}=FQ_n$; hence
the enhanced hypercube is a generalization of the folded hypercube.
The graphs shown in Fig.~\ref{f1} are $Q_{3,3}$ and $Q_{4,3}$, where
the hypercube edges and $3$-complementary edges are represented by
solid lines and dashed lines, respectively.

\begin{figure}[h!]
\label{fig:EQ1}
\begin{pspicture}(-0.6,-.0)(3,4)
\psset{radius=.08, unit=.8}

\Cnode(1,1){000}\rput(.6,0.75){\tiny000}
\Cnode(3.5,1){010}\rput(3.8,.75){\tiny010}
\Cnode(1,3.5){100}\rput(.5,3.5){\tiny100}
\Cnode(3.5,3.5){110}\rput(3.4,3.8){\tiny110}

\Cnode(2.3,1.8){001}\rput(2.5,1.6){\tiny001}
\Cnode(4.8,1.8){011}\rput(5.25,1.98){\tiny011}
\Cnode(2.3,4.3){101}\rput(2.3,4.6){\tiny101}
\Cnode(4.8,4.3){111}\rput(4.8,4.6){\tiny111}

\ncline[linewidth=1.pt]{000}{100} \ncline[linewidth=1.pt]{001}{011}
\ncline[linewidth=1.pt]{101}{111} \ncline[linewidth=1.pt]{011}{111}

\ncline[linewidth=1.pt]{001}{101} \ncline[linewidth=1.pt]{000}{010}
\ncline[linewidth=1.pt]{100}{110} \ncline[linewidth=1.pt]{010}{110}

\ncline[linewidth=1.pt]{010}{011} \ncline[linewidth=1.pt]{110}{111}
\ncline[linewidth=1.pt]{100}{101} \ncline[linewidth=1.pt]{000}{001}

\ncline[linestyle=dashed, dash=3pt
2pt]{101}{010}\ncline[linestyle=dashed, dash=3pt 2pt]{100}{011}

\ncline[linestyle=dashed, dash=3pt
2pt]{000}{111}\ncline[linestyle=dashed, dash=3pt 2pt]{001}{110}

\rput(2.3,0.1){\scriptsize (a) $Q_{3,3}$}
\end{pspicture}
\begin{pspicture}(-1.6,-.0)(6,4)
\psset{radius=.08, unit=.8}

\Cnode(1,1){0000}\rput(.6,0.75){\tiny0000}
\Cnode(3.5,1){0010}\rput(3.8,.75){\tiny0010}
\Cnode(1,3.5){0100}\rput(.5,3.5){\tiny0100}
\Cnode(3.5,3.5){0110}\rput(3.36,3.75){\tiny0110}

\Cnode(2.3,1.8){0001}\rput(2.5,1.6){\tiny0001}
\Cnode(4.8,1.8){0011}\rput(5.25,1.98){\tiny0011}
\Cnode(2.3,4.3){0101}\rput(2.3,4.6){\tiny0101}
\Cnode(4.8,4.3){0111}\rput(4.8,4.6){\tiny0111}

\ncline[linewidth=1.pt]{0000}{0100}
\ncline[linewidth=1.pt]{0001}{0011}
\ncline[linewidth=1.pt]{0101}{0111}
\ncline[linewidth=1.pt]{0011}{0111}

\ncline[linewidth=1.pt]{0001}{0101}
\ncline[linewidth=1.pt]{0000}{0010}
\ncline[linewidth=1.pt]{0100}{0110}
\ncline[linewidth=1.pt]{0010}{0110}

\ncline[linewidth=1.pt]{0010}{0011}
\ncline[linewidth=1.pt]{0110}{0111}
\ncline[linewidth=1.pt]{0100}{0101}
\ncline[linewidth=1.pt]{0000}{0001}

\ncline[linestyle=dashed, dash=3pt
2pt]{0101}{0010}\ncline[linestyle=dashed, dash=3pt 2pt]{0100}{0011}

\ncline[linestyle=dashed, dash=3pt
2pt]{0000}{0111}\ncline[linestyle=dashed, dash=3pt 2pt]{0001}{0110}

\Cnode(6,1){1010}\rput(5.6,0.75){\tiny1010}
\Cnode(8.5,1){1000}\rput(8.8,.75){\tiny1000}
\Cnode(6,3.5){1110}\rput(6.7,3.65){\tiny1110}
\Cnode(8.5,3.5){1100}\rput(8.4,3.8){\tiny1100}

\Cnode(7.3,1.8){1011}\rput(7.5,1.6){\tiny1011}
\Cnode(9.8,1.8){1001}\rput(10.25,1.8){\tiny1001}
\Cnode(7.3,4.3){1111}\rput(7.3,4.6){\tiny1111}
\Cnode(9.8,4.3){1101}\rput(9.8,4.6){\tiny1101}

\ncline[linewidth=1.pt]{1000}{1100}
\ncline[linewidth=1.pt]{1001}{1011}
\ncline[linewidth=1.pt]{1101}{1111}
\ncline[linewidth=1.pt]{1011}{1111}

\ncline[linewidth=1.pt]{1001}{1101}
\ncline[linewidth=1.pt]{1000}{1010}
\ncline[linewidth=1.pt]{1100}{1110}
\ncline[linewidth=1.pt]{1010}{1110}

\ncline[linewidth=1.pt]{1010}{1011}
\ncline[linewidth=1.pt]{1110}{1111}
\ncline[linewidth=1.pt]{1100}{1101}
\ncline[linewidth=1.pt]{1000}{1001}

\ncline[linewidth=1.pt]{1010}{0010}
\ncline[linewidth=1.pt]{1011}{0011}
\ncline[linewidth=1.pt]{1110}{0110}
\ncline[linewidth=1.pt]{1111}{0111}

\ncline[linestyle=dashed, dash=3pt
2pt]{1101}{1010}\ncline[linestyle=dashed, dash=3pt 2pt]{1100}{1011}

\ncline[linestyle=dashed, dash=3pt
2pt]{1000}{1111}\ncline[linestyle=dashed, dash=3pt 2pt]{1001}{1110}

\nccurve[angleA=17,angleB=163]{0101}{1101}
\nccurve[angleA=17,angleB=163]{0100}{1100}
\nccurve[angleA=17,angleB=163]{0001}{1001}
\nccurve[angleA=17,angleB=163]{0000}{1000}

\rput(4.8,0.1){\scriptsize (b) $Q_{4,3}$}

\end{pspicture}
\caption{\label{f1} Enhanced hypercubes $Q_{3,3}$ and $Q_{4,3}$}
\end{figure}
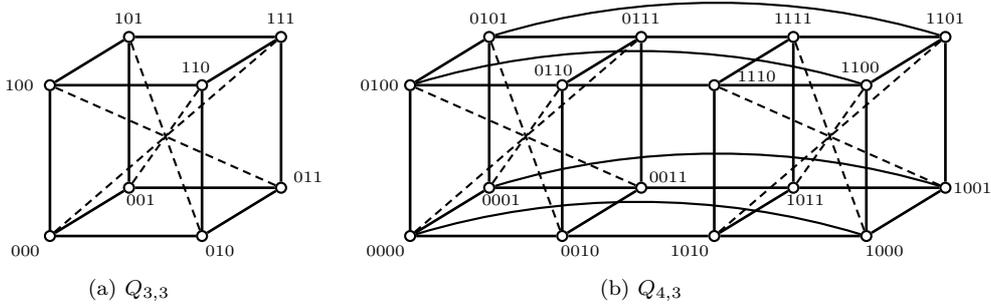

A graph $G$ is {\it vertex-transitive} if for any $u,v\in V(G)$
there is some $\sigma \in Aut(G)$, the automorphism group of $G$,
such that $\sigma(u)=v$; it is {\it edge-transitive} if for any $xy,
uv \in E(G)$ there is some $\sigma \in Aut(G)$ such that
$\{\sigma(x), \sigma(y)\} = \{u, v\}$. The hypercube $Q_n$ and
folded hypercube $FQ_{n}$ are vertex-transitive and edge-transitive,
and the enhanced hypercube $Q_{n,k}$ is vertex-transitive but not
edge-transitive when $k<n$~\cite{mx2010, x01, ycpc2014}.

The {\it Cartesian product} $G \Box H$ of graphs $G$ and $H$ is
the graph with vertex set $V(G)\times V(H)$, in which vertices
$(u,v)$ and $(u',v')$ are adjacent whenever $uu'\in E(G)$ and
$v=v'$, or $u=u'$ and $vv'\in E(H)$. By the definition of
$Q_{n,k}$, we have $Q_{n,k}=Q_{n-k}\Box FQ_{k}$. Although
$Q_{n,k}$ is not edge-transitive when $2\leq k< n$, we have the
following property.

\begin{prop}\label{prop1}
Permuting the first $k$ positions and/or permuting the last $n-k$ positions
in the names of the vertices of $Q_{n,k}$ does not change the graph.
\end{prop}

\begin{proof}
Exchanging the $i$th bit and $j$th bit among the first $k$ or among
the last $n-k$ preserves the adjacency relation, since the number of
coordinates in which two vertices differ is not changed
by exchanging such coordinates.  An arbitrary permutation is obtained
by a succession of such exchanges.
\end{proof}

\begin{prop}\label{prop2}
Given $u,v\in V(Q_{n,k})$, let $r$ and $s$ be the numbers of positions
in which $u$ and $v$ differ among the first $k$ and last $n-k$ positions,
respectively.  The distance between $u$ and $v$ is computed by
$d(u,v)=s+\min\{r,k-r+1\}$.
\end{prop}

\begin{proof}
The distance is the minimum number of steps to change $u$ into $v$.  All steps
change one bit, except that a $k$-complementary edge changes the first $k$
bits.  If no $k$-complementary edge is used, then the number of steps is at
least the Hamming distance, and this suffices.  For this reason, a shortest
path uses at most one $k$-complementary edge.  If a $k$-complementary edge is
used, then the $k-r$ positions in which $u$ and $v$ agree among the first $k$
must be changed individually.
\end{proof}

Proposition~\ref{prop2} immediately yields
$d(Q_{n,k})=(n-k)+\lceil\frac {k}{2}\rceil$, which equals
$n-\lfloor \frac {k}{2}\rfloor$.
This was observed by Tzeng and Wei~\cite{tw1991}, along with
an algorithm for finding shortest paths joining vertices.
Note that if $u$ and $v$ differ in more than $\lceil \frac {k}{2}\rceil$
positions among the first $k$, then every shortest
$u,v$-path contains exactly one $k$-complementary edge.

\section{Construction of Paths}

Many properties of interconnection networks were investigated by
different construction methods of paths~\cite{chf14, l15}. In this
section we will prove our main results by constructing disjoint
paths of bounded length joining any two vertices in $Q_{n,k}$.

Let $P$ be a path $u_0\rightarrow u_1\rightarrow \cdots \rightarrow
u_{\ell-1}\rightarrow u_{\ell}$ from the vertex $u_0$ to a vertex
$u_\ell$ in $Q_{n,k}$. The path $P$ traverses the edges
$u_0u_1,u_1u_2,\ldots,u_{\ell-1}u_\ell$. Since the edge $u_{j-1}u_j$
is in $E_{d_j}$ for some $d_j\in\{0,1,\ldots,n\}$, and every vertex
is incident with exactly one edge in $E_{d_j}$, we can represent the
path $P$ from $u_0$ by the list $(d_1,\ldots, d_\ell)$, where $d_j$
indicates the type of the edge joining $u_{j-1}$ and $u_j$. For
example, in $Q_{5,3}$ the path originating from $00000$ determined
by $(2,0,5)$ is $00000\xrightarrow{2} 000{\bf 1}0 \xrightarrow{0}
00{\bf 101} \xrightarrow{5} {\bf 1}0101$. Note that the length of a
path determined by a list $I$ is the number of elements in $I$.

We use the following two lemmas to construct disjoint paths.
A {\it proper segment} of a list is a string of consecutive elements
that is not the full list.

\begin{lemma}\label{lem4}
No two cyclic permutations of a list of distinct elements have a
common proper initial segment.
\end{lemma}

\begin{proof}
Let $I_1=(d_1,d_2, \ldots, d_\ell)$ and $I_j=(d_j,\ldots,d_\ell,
d_1,\ldots, d_{j-1})$ for $2\leq j \leq \ell$.

Consider $j$ and $j'$ with $1\leq j'< j \leq \ell$.  A proper initial segment of
$I_j$ for $j\geq 2$ contains $d_j$ and not $d_{j-1}$, but this
is not true for any initial segment of $I_{j'}$, since when $d_j$ is not
the first element, $d_{j-1}$ also occurs in any initial segment containing
$d_j$.  Therefore, $I_{j}$ and $I_{j'}$ have no common proper initial segment.
\end{proof}

\begin{lemma}\label{distinct}
Let $S$ be a set of $\ell$ distinct elements of $\{0,1,\dots,n\}$.  Let
$\VEC I1m$ be orderings of $S$ such that no two have a common proper initial
segment.  If $\{0,1,\ldots,k\}$ are not all in $S$, then for any vertex $u$ in
$Q_{n,k}$, the lists $I_1, \ldots, I_m$ determine disjoint paths to a single
vertex.
\end{lemma}

\begin{proof}
Let $S=\{d_1,\ldots, d_\ell\}$ and $[k]=\{1,\ldots,k\}$. If $0\notin S$, then
each path reaches the vertex $v$ that differs from $u$ in the positions of $S$.
If $0\in S$, then each path reaches the vertex $v$ that differs from $u$ in the
positions of $([k]-S)\cup (S-\{0\}-[k])$.

If two subsets of $S$ both contain $0$ or both omit $0$, then they produce
paths to the same vertex only if they are the same set.  If $0\in T\subseteq S$
and $0\notin T'\subseteq S$, then $T$ and $T'$ produce paths to the same vertex
only if they agree outside $[k]$ and intersect $[k]$ in complementary subsets.
By the hypothesis that $0,1,\dots,k$ are not all present in $S$, this cannot
occur.

Therefore, the paths from $u$ determined by $I_j$ and $I_{j'}$ have a common
internal vertex if and only if they have a common proper initial segment.
\end{proof}

\begin{lemma}\label{lem5}
Let $I=(d_1,\ldots,d_\ell)$ with all $d_i$ distinct and in $\{0,1,\ldots,n\}$.
If $\{0,1,\ldots,k\}$ are not all in $I$, then for any vertex $u$ in $Q_{n,k}$,
the $\ell$ cyclic permutations of $I$ determine disjoint paths to a single
vertex.
\end{lemma}
\begin{proof}
The conclusion follows immediately from Lemmas~\ref{lem4} and~\ref{distinct}.
\end{proof}

When $G$ is the complete graph $K_n$ with $n\ge3$, we have $D_\omega(G)=1$ but
$d_\omega(G)=2$ for $2\le\omega\le n-1$.  Since $Q_{2,2}=K_4$, we consider
$Q_{n,k}$ with $n\ge 3$ and $2\leq k\leq n$.

\begin{theorem}\label{lem1}
For any two distinct vertices $u$ and $v$ in $Q_{n,k}$ with $n\ge3$ and
$2\leq k\leq n$, there exist $n+1$ disjoint $u,v$-paths of length at most
$d(Q_{n,k})+1$, such that at least $n-\lfloor\frac{k}{2}\rfloor-1$ of the
paths have length at most $d(Q_{n,k})$.
\end{theorem}

\begin{proof}

The vertex transitivity of $Q_{n,k}$ and Proposition~\ref{prop1} allow us to
assume $u=0^{n}$ and $v=0^{n-k-j}1^{j}0^{k-i}1^{i}$, where $0\leq i \leq k$ and
$0\leq j \leq n-k$. We will construct the desired paths.  Let
$r=\min\{i+j,k-i+j+1\}$.  Recall that $d(Q_{n,k})=n-\lfloor\frac k2\rfloor$.
We consider two cases according to the relationship between $r$ and
$d(Q_{n,k})$.

\vskip6pt

{\it Case 1: $r< d(Q_{n,k})$.} We first specify a list $I$ of length
$r$, in two cases (Table\ref{tab1}). Note that $r=i+j$ when $i\le
k-i$ and $r=k-i+j+1$ when $i>k-i$.

\begin{table}[ht]
\caption{Two cases for $r< d(Q_{n,k})$.} \label{tab1} \centering
\begin{tabular}[h]{|c|l||c|}
\hline
Case & Condition & $I$ \\
\hline
A&$i\le k-i$& $(1,\dots,i,k+1,\dots,k+j)$\\
B&$i> k-i$  & $(0,i+1,\dots,k,k+1,\dots,k+j)$\\
\hline
\end{tabular}
\end{table}
Not all of $0,1,\dots,k$ appear in $I$, since having $0$ requires $i>k/2$
(Case B), and then also having $1$ requires $i=0$, a contradiction.  Hence
Lemma~\ref{lem5} applies, so the $r$ cyclic permutations of $I$ determine $r$
disjoint $u,v$-paths of length $r$.

The remaining $n-r+1$ paths, with length $r+2$, are specified by adding one of
$\{0,1,\dots,n\}-I$ at both the beginning and the end of $I$.  Let $I'$ be the
list obtained by adding $h$, and let $P$ be the path from $u$ determined by
$I'$.

If $h>k$ or if $h\in[k]$ and $0\notin I$, then $u$ and $v$ agree in position
$h$, and $P$ is the only path in the constructed set containing vertices that
differ from them in position $h$.  Furthermore, all internal vertices of $P$
differ from $u$ and $v$ in position $h$.

If $h=0$, then $i\le k-i$ (Case A).  All internal vertices of $P$ differ from
$u$ and $v$ in positions $k$ and $k-1$, and no other path has any such vertices.

If $h\in[k]$ and $0\in I$, then $1\le h\le i$ (Case B).  The first vertex of
$P$ after $u$ differs from $u$ only in position $h$.  The next vertex disagrees
with $u$ on all of positions $1,\dots,i$ except $h$, and this remains true of
all other internal vertices of $P$, because $I$ contains no element of
$\{1,\dots,i\}$.  All the other paths in the construction have no vertices
satisfying either of these conditions.

Since $r<d(Q_{n,k})$, all the paths have length at most $d(Q_{n,k})+1$;
in fact, all have length at most $d(Q_{n,k})$ unless $r=d(Q_{n,k})-1$.
In this case there are $r$ paths of length $r$, which suffices since
$r=d(Q_{n,k})-1=n-\fktwo-1$.

\vskip6pt

{\it Case 2: $r\ge d(Q_{n,k})$.} Let $s=d(Q_{n,k})$.  Since
$r=\min\{i+j,k-i+j+1\}$, we have $i+j\ge s$ and $k-i+j+1\ge s$, so
$k+2j+1\ge 2s$.  If $j\le n-k-1$, then $2n-k-1\ge 2s=2n-2\fktwo$,
which is impossible.  Hence $j=n-k$. With $j=n-k$, we have $i\ge
\cktwo$ and $k+1-i\ge\cktwo$, so $\cktwo\le i\le \fktwo+1$. When
$i=\cktwo$, we have $r=i+j$; when $i=\fktwo+1$, we have $r=k-i+j+1$.
Both cases apply when $k$ is odd. In either case,
$r=n-\fktwo=d(Q_{n,k})$, and we define three
lists(Table~\ref{tab2}).

\begin{table}[ht]\caption{Two cases for $j=n-k$ and $\cktwo\le i\le \fktwo+1$.} \label{tab2}
\centering
\begin{tabular}[h]{|c|l||c|c|c|}
\hline
Case &\quad $i$ & $I$ & $J$ & $I'$ \\
\hline
A&$\cktwo$&   $(1,\dots,i,k+1,\dots,n)$& $(k+1,\ldots,n)$& $(0,i+1,\dots,k)$\\
B&$\fktwo+1$& $(0,i+1,\dots,k,k+1,\dots,n)$& $(k+1,\ldots,n)$& $(1,\dots,i)$\\
\hline
\end{tabular}
\end{table}

Since $j=n-k$, in each case $I$ has length $r$.  By Lemma~\ref{lem5}, the
cyclic permutations of $I$ yield $r$ disjoint $u,v$-paths of length $r$. Since
$r=n-\fktwo=d(Q_{n,k})$, this yields enough paths of length at most
$d(Q_{n,k})$.  Let $T=\{i+1,\dots,k\}$ in Case A, $T=\{1,\dots,i\}$ in Case B.
Every vertex in each of these paths is constant in the positions of $T$
(all-$0$ or all-$1$), and in fact all-$0$ in Case A.

Since $r=n-\fktwo$, we only need to find $\fktwo+1$ more paths of length at
most $d(Q_{n,k})+1$.  Note that $I'$ has length $\fktwo+1$.  Form $\fktwo+1$
lists by inserting $J$ after the first element of each cyclic permutation of
$I'$.  The first and last lists are $(0,J,i+1,\dots,k)$ and
$(k,J,0,i+1,\dots,k-1)$ in Case A, $(1,J,2,\dots,i)$ and $(i,J,1,\dots,i-1)$ in
Case B.  Each of these lists has length $n-\cktwo+1$, which is at most
$d(Q_{n,k})+1$.

Each of these lists is an ordering of a single set of elements.  By an argument
like that of Lemma~\ref{lem4}, they have no common proper initial segments.
Since they also do not contain all of $\{0,1,\dots,k\}$, by
Lemma~\ref{distinct} these paths are disjoint.

In Case A, each internal vertex on each of these paths is not all $0$ in the
positions of $T$.  In Case B, each internal vertex on each of these paths
has between $1$ and $|T|-1$ nonzero positions in $T$.  Hence these paths are
disjoint from the earlier paths.
\end{proof}

\section{Consequences}

From Theorem~\ref{lem1} and the definition of wide diameter, we immediately
obtain an upper bound on $d_{\omega}(Q_{n,k})$.

\begin{corollary} \label{cor1} If $n\ge3$ and $2\leq k \leq n$, then
 $$
 d_{\omega}(Q_{n,k}) \le \left\{\begin{array}{ll}
 d(Q_{n,k})\ & {\rm for}\ 1 \leq \omega < n-\fktwo,\\
 d(Q_{n,k})+1\ & {\rm for}\ n-\lfloor\frac{k}{2}\rfloor \leq \omega \leq n+1.
 \end{array}\right.
 $$
\end{corollary}
\begin{proof}
When $\omega<n-\fktwo$, Theorem~\ref{lem1} provides at least $\omega$ disjoint
paths with length at most $d(Q_{n,k})$ joining any two vertices in $Q_{n,k}$.
When $\omega\le n+1$, it provides at least $\omega$ such paths with length
at most $d(Q_{n,k})+1$.
\end{proof}

We next give a lower bound on the fault diameter $D_\omega(Q_{n,k})$.

\begin{lemma} \label{lem2} Fix $n\ge3$.  If $2\leq k \leq n$
and $n-\lfloor\frac{k}{2}\rfloor \leq \omega \leq n+1$, then
$$D_{\omega}(Q_{n,k}) \geq d(Q_{n,k})+1.$$
\end{lemma}

\begin{proof}
Since $D_\omega(G)$ is nondecreasing in $\omega$,
proving $D_{n-\lfloor\frac{k}{2}\rfloor}(Q_{n,k}) \geq d(Q_{n,k})+1$
is sufficient.
Let $u=0^{n}$ and $v=1^{n-k}0^{k-i}1^{i}$, where $i=\cktwo-1$.  Note that $v$
has $1$s in $n-\fktwo-1$ positions.  Let $W$ be the set of neighbors of $u$
whose single $1$ occurs in a position where $v$ has a $1$, so $|W|=n-\fktwo-1$.
On any $u,v$-path in $Q_{n,k}-W$, the neighbor of $u$ has a single $1$ in
a position among $i+1,\dots,k$ or is $0^{n-k}1^k$.

By Proposition~\ref{prop2}, the distance between $v$ and a neighbor $u'$ of $u$
not in $W$ is $n-k+\min\{i+1,k-i\}$.  Since $i=\cktwo-1$, the distance is
$n-\fktwo$, which equals $d(Q_{n,k})$.  Hence every $u,v$-path in
$Q_{n,k}-W$ has length at least $d(Q_{n,k})+1$.
\end{proof}

\begin{theorem} \label{T3} For $3\leq n$ and $2\leq k \leq n$,
$$D_\omega(Q_{n,k})=d_\omega(Q_{n,k})=\begin{cases}
    d(Q_{n,k}) & \textrm{for $1 \leq \omega < n-\lfloor\frac{k}{2}\rfloor$;}  \\
    d(Q_{n,k})+1 & \textrm{for $n-\lfloor\frac{k}{2}\rfloor \leq \omega \leq n+1$.}
    \end{cases}$$

\end{theorem}

\begin{proof} Since $d(Q_{n,k}) \leq D_\omega(Q_{n,k})\leq d_\omega(Q_{n,k})$
for $1\leq \omega \leq n+1$, Corollary~\ref{cor1} yields
$D_\omega(Q_{n,k})=d_\omega(Q_{n,k})=d(Q_{n,k})$ for
$1 \leq \omega <n-\lfloor\frac{k}{2}\rfloor$.

For $n-\lfloor\frac{k}{2}\rfloor \leq \omega \leq n+1$,
Corollary~\ref{cor1} and Lemma~\ref{lem2} yield
$D_\omega(Q_{n,k})=d_\omega(Q_{n,k})=d(Q_{n,k})+1$.
\end{proof}

Theorem~\ref{T3} shows that the fault diameter $D_\omega(Q_{n,k})$
equals the wide diameter $d_\omega(Q_{n,k})$ for the enhanced
hypercubes $Q_{n,k}$.  More importantly, they equal the traditional diameter
when $1\leq \omega < n-\lfloor\frac{k}{2}\rfloor$, and they exceed it
only by $1$ when $n-\lfloor\frac{k}{2}\rfloor \leq \omega \leq n+1$.
Thus, the resilience of enhanced hypercubes is similar to that of hypercubes,
which increases the appeal of enhanced hypercubes.

\section*{Acknowledgements}

This work was supported partially by NSFC (Nos. 11101378, 11571044).
Research of D.B. West was supported by 1000 Talent Plan, State
Administration of Foreign Experts Affairs, China.

\end{document}